\documentclass[journal,10pt, twocolumn]{IEEEtran}
\usepackage{amsmath}
\usepackage{amsfonts}
\usepackage{mathrsfs}
\usepackage{amssymb}
\usepackage{mathbbold}
\usepackage[keeplastbox]{flushend}
\usepackage{cases}
\usepackage{stmaryrd}
\usepackage{mathrsfs}
\usepackage{bbm}
\usepackage[dvips]{graphicx}
\usepackage{CJK}
\usepackage{amsmath}
\usepackage{flushend}
\usepackage{algorithm}
\usepackage{algorithmic}
\usepackage{multirow}
\usepackage{array}
\usepackage{algorithm}
\usepackage{algorithmic}
\usepackage{graphicx}
\usepackage{subfigure}
\usepackage{color}
\usepackage{amsthm}
\usepackage{cite}
\usepackage{xcolor}
\newtheorem{pp}{Proposition}

\newtheorem{corollary}{Corollary}

\begin{document}

\title{Robust Beamforming and Jamming for Enhancing the Physical Layer Security of Full Duplex Radios}

\author{Zhengmin Kong, Shaoshi Yang,~\IEEEmembership{Senior Member,~IEEE}, Die Wang, Lajos~Hanzo,~\IEEEmembership{Fellow,~IEEE}

\thanks{This work is financially supported by the National Natural Science Foundation of China (NSFC) (Grant No. 61801518) and the Hubei Provincial Natural Science Foundation of China (Grant No. 2017CFB661). (\textit{Corresponding author: Shaoshi Yang})

Z. Kong and D. Wang are with the Automation Department,School of Electrical Engineering and Automation, Wuhan University, Wuhan 430072, China (e-mail:
zmkong@whu.edu.cn, wangdie1995@whu.edu.cn).

S. Yang is with the School of Information and Communication Engineering, Beijing University of Posts and Telecommunications, and with the Key Laboratory of Universal Wireless Communications, Ministry of Education,  Beijing 100876, China (e-mail: shaoshi.yang@ieee.org).

L. Hanzo is with the School of Electronics and Computer Science, University of Southampton, Southampton SO17 1BJ, U.K. (e-mail: lh@ecs.soton.ac.uk)}% <-this % stops a space
\thanks{}}

\markboth{}
{}
%{Shell \MakeLowercase{\textit{et al.}}: Bare Demo of IEEEtran.cls for Journals}

\maketitle

\begin{abstract}
In this paper, we investigate the physical layer security of a full-duplex base station (BS) aided system in the worst case, where an uplink transmitter (UT) and a downlink receiver (DR) are both equipped with a single antenna, while a powerful eavesdropper is equipped with multiple antennas. For securing the confidentiality of signals transmitted from the BS and UT, an artificial noise (AN) aided secrecy beamforming scheme is proposed, which is robust to the realistic imperfect state information of both the eavesdropping channel and the residual self-interference channel. Our objective function is that of maximizing the worst-case sum secrecy rate achieved by the BS and UT, through jointly optimizing the beamforming vector of the confidential signals and the transmit covariance matrix of the AN. However, the resultant optimization problem is non-convex and non-linear. In order to efficiently obtain the solution, we transform the non-convex problem into a sequence of convex problems by adopting the block coordinate descent algorithm. We invoke a linear matrix inequality for finding its Karush-Kuhn-Tucker (KKT) solution. In order to evaluate the achievable performance, the worst-case secrecy rate is derived analytically. Furthermore, we construct another secrecy transmission scheme using the projection matrix theory for performance comparison. Our simulation results show that the proposed robust secrecy transmission scheme achieves substantial secrecy performance gains, which verifies the efficiency of the proposed method.

%the substantial secrecy performance gains can be achieved and the performance gains increase with the increasing transmit power.
\end{abstract}


\IEEEpeerreviewmaketitle

\section{Introduction}
Full-duplex (FD) communication has the potential of doubling the
spectral efficiency compared to its half-duplex counterpart. The main
challenge facing FD systems is the increased self-interference imposed
by the leakage from the transmitted signals to the signals received at
the FD node. Sophisticated techniques, such as TX/RX antenna
separation and isolation, as well as digital/analog/propagation domain
interference cancellation, have been proposed for combating the
self-interference \cite{FD_radio, InbandFullduplex}. Thus the design
of a single-channel FD system becomes feasible.

On the other hand, exploiting artificial noise (AN) \cite{AN} has been
widely recognized as an effective technique of improving the physical
layer security (PLS) \cite{Yang_JSAC_2018, Yang_JSAC, Kong_TIFS,
  Cao_TIFS, Zou_TVT_2014, Zou_TVT_2017} of wireless communications. By
transmitting AN in addition to the confidential signals, the quality
of the eavesdropping channel is degraded. However, existing AN based
PLS schemes cannot be used for securing uplink (UL) transmissions in
conventional FDD/TDD cellular systems \cite{Jin_uplink_secrecy,
  Jin_uplink_secrecy_scheduling}, because i) the base stations (BSs)
cannot transmit AN while receiving co-channel desired signals; ii) an
UL transmitter (UT) typically has a single transmit antenna that can
only be used for transmitting desired information-bearing signals to a
BS, and no more antennas are available for transmitting AN. Notably,
FD techniques provide a solution to such a dilemma, since an FD
transceiver is capable of receiving the confidential signals, while
simultaneously transmitting AN in order to jam potential eavesdroppers
(Eves). The authors of \cite{FD_statistics_CSI} proposed a user-grouping-based fractional time model relying either on perfect channel state information (CSI) or on statistical CSI. Assuming perfect self-interference suppression, the authors of \cite{FullDuplexBaseStation} studied the joint design of the confidential signal and AN for maximizing the instantaneous secrecy rate achieved by the legitimate downlink (DL) receiver (DR), while keeping the received secrecy rate at the FD BS above a predefined target.  By contrast, assuming imperfect self-interference
cancellation, the authors of \cite{FullDuplexBaseStation2} studied the
joint design of the confidential signal and AN for minimizing the
total transmit power, while keeping the achievable secrecy rate of the
FD BS and of the DR above a predefined target. Furthermore, the authors of \cite{FD_relay_perfectCSI} proposed a pair of relay-assisted secure protocols by relying on the FD capability of an orthogonal frequency division multiple access (OFDMA) system.
%Furthermore, the authors of \cite{FD_statistics_CSI} proposed a new and very efficient user grouping-based fractional time model assuming perfect CSI and the statistics CSI, where the FD-BS uses a fraction of the time block to serve near DL users and far UL users, and the remaining fractional time to serve other users.

However, none of the aforementioned contributions considered the scenario of realistic imperfect CSI, even though only imperfect eavesdropping channel information might be gleaned by the BS - if any at all - since typically the Eves sends no training signals to the BS.

Although the authors of \cite{FD_Relay_incomplete_ECSI,RobustSecrecy} analyzed the secrecy outage probability in an imperfect eavesdropper CSI knowledge scenario for a FD relay system, the eavesdropper was only equipped with a single antenna and the source was unable to communicate with the destination directly.

Hence, more realistic practical problems have to be considered. More specifically, the security in the DL and UL have to be considered simultaneously, and the sum secrecy rate of the DL and UL should be optimized in the face of a ``sophisticated/strong'' eavesdropper having imperfect CSI\footnote{The eavesdropper is said to be ``sophisticated/strong'', for example when it has multiple antennas or when the number of eavesdropper antennas is higher than that of each legitimate UT and DR, or alternatively, when it sends no training signals to the BS.}.

In contrast to \cite{FD_statistics_CSI, FullDuplexBaseStation, FullDuplexBaseStation2, FD_relay_perfectCSI,  FD_Relay_incomplete_ECSI}, in this paper, our new contribution is that we jointly design the beamforming vector of the confidential signal and the covariance matrix of the AN, in order to maximize the worst-case\footnote{In this treatise, we consider the worst-case scenario, where the UT and DR are equipped with a single antenna, while the ``sophisticated/strong'' eavesdropper is equipped with multiple antennas. Furthermore, when only imperfect eavesdropping channel knowledge is available to the BS. In such cases, the capacity of the eavesdropper is potentially higher than that of the legitimate user. Naturally, this scenario is more challenging than the case where the legitimate users have the same capability as the eavesdropper.} sum secrecy rate achieved by the FD BS and the UT, under the realistic assumption that only imperfect state information of the residual self-interference channel and of the eavesdropping channel is available to the FD BS.

To the best of our knowledge, there is no open literature addressing the simultaneous optimization of the sum secrecy rate of the DL and UL in the face of realistic CSI error concerning a ``sophisticated/strong'' eavesdropper's channel in the above-mentionedf worst-case FD scenario. Since the worst-case optimization problem is considered, in the paper, we adopt the deterministic model of\cite{RFullDuplex,Huang_TSP} to characterize the imperfect channel state information (CSI).  Because the resultant objective function considered is non-convex, the global optimum is hard to obtain. As a low-complexity suboptimal method, the block coordinate descent (BCD) algorithm of \cite[Section 2.7]{NolinearProgramming} is adopted to transform the non-convex and non-linear problem into a sequence of convex problems. In other words, without using any approximation method such as the classic semi-definite relaxation (SDR) technique, we transform our non-convex and non-linear problem into a more well-behaved form and conceive an efficient algorithm for finding its Karush-Kuhn-Tucke(KKT) solution. Then the locally optimal solution of the original problem is obtained. Our analysis and simulation results reveal valuable insights into how several relevant design parameters affect the achievable secrecy performance of the system considered; compared to projection matrix theory, the proposed robust secrecy transmission scheme is capable of achieving substantial secrecy performance gains.
%Based on the deterministic model, the authors of \cite{RFullDuplex} proposed an approach to maximize the achievable sum secrecy rate of two-way secure communications under the worst-case channel condition.
%This approach has been widely used in the design of AN  aided secrecy transmission in half-duplex communications \cite{Huang_TSP}, which can guarantee the absolute robustness, since the optimization objective is the worst-case secrecy performance.

The rest of this paper is organized as follows. In Section II, the system model and the problem formulation are described. In Section III, we transform the original non-convex problem into a sequence of convex problems and invoke an iterative BCD algorithm for finding its Karush-Kuhn-Tucke(KKT) solution. Our simulation results are provided in Section IV, and the conclusions are offered in Section V.

\emph{Notation:} $(\cdot)^T$, $(\cdot)^*$ and $(\cdot)^H$ denote the transpose, conjugate and conjugate transpose, respectively.  $(\cdot)^{-1}$, $\mathrm{tr}(\cdot)$, and $||\cdot||_F$ denote the inverse, trace, and Frobenius norm of a matrix, respectively. $\ln(\cdot)$ denotes the natural logarithm, and $\otimes$ represents the Kronecker product.
\section{System Model and Problem Formulation}
\subsection{Transmission Model}
\begin{figure}[!t]
\centering
\includegraphics[width=1.8in]{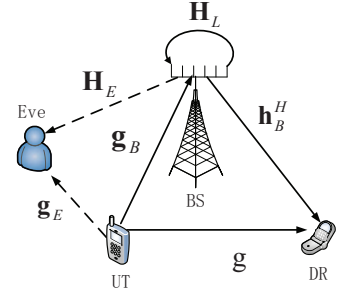}
\caption{A cellular system composed of a multi-antenna FD BS, a single-antenna UT, a single-antenna DR, and a multi-antenna Eve.}
\label{SystemModel}
\end{figure}
We consider an FD cellular system composed of an FD BS, a single-antenna legitimate DR, a single-antenna legitimate UT and a powerful multi-antenna Eve, as illustrated in Fig. \ref{SystemModel}. The UT and BS are scheduled for UL and DL transmission respectively, in the same frequency band.
% A multiple-antenna Eve can be regarded as multiple unauthorized users, which share their antennas and perform joint
%processing to wiretap the confidential information.
We assume that the Eve has $N_E$ receiving antennas, while the FD BS has $N_T$ transmitting antennas and $N_R$ receiving antennas to facilitate simultaneous transmission and reception.  Additionally, we assume $N_T\geq N_E+1$, so that the BS has sufficient degrees of freedom to transmit the AN.
As shown in Fig. \ref{SystemModel}, the DL channels from the BS to the DR and to the Eve are denoted by
$\mathbf{h}^H_{B} \in\mathbb{C}^{1 \times N_T}$ and $\mathbf{H}_{E} \in\mathbb{C}^{N_E\times N_T}$, respectively, while the UL channels from the UT to the BS, the Eve and the DR are denoted by $\mathbf{g}_B \in\mathbb{C}^{N_R\times 1}$,
$\mathbf{g}_E \in\mathbb{C}^{N_E\times 1}$, and $g$, respectively.
%$\mathbf{h}^H_{B}\triangleq\sqrt{\beta_{{BD}}}\mathbf{f}^H_{BD}$ and $\mathbf{H}_{E}\triangleq\sqrt{\beta_{{BE}}}\mathbf{F}_{BE}$, respectively, where
%$\mathbf{f}^H_{BD}\in\mathbb{C}^{1\times N_T}$ and $\mathbf{F}_E\in\mathbb{C}^{N_E\times N_T}$ denote the small-scale fading effects, and $\beta_{{xy}}$ defines the corresponding large-scale fading effect. We consider the standard path loss model $\beta_{{xy}}\triangleq d_{xy}^{-\alpha}$, where $d_{xy}$ is the corresponding distance, and $\alpha$ is the path-loss exponent. Accordingly, the uplink channel vectors from UT to BS, Eve and DR are denoted by $\mathbf{g}_B\triangleq\sqrt{\beta_{{UB}}}\mathbf{f}_{UB}$,
%$\mathbf{c}_E\triangleq\sqrt{\beta_{{UE}}}\mathbf{f}_{UE}$, and
%$t\triangleq\sqrt{\beta_{{UD}}}{f}_{UD}$, where $\mathbf{f}_{UB}\in\mathbb{C}^{N_T\times 1}$, $\mathbf{f}_{UE}\in\mathbb{C}^{N_E\times 1}$, and ${f}_{UD}$ denotes the small-scale fading effects.
Finally, $\mathbf{H}_L \in\mathbb{C}^{N_R\times N_T}$ denotes the residual self-interference channel due to the imperfect self-interference cancellation \cite{InbandFullduplex}.

In the DL transmission, the BS transmits the confidential information-bearing signals together with the AN to protect the confidential information from wiretapping. The baseband signal from the BS can be expressed as
\begin{align}
\mathbf{x}_b\triangleq \mathbf{v} s + \mathbf{n}_a,
\end{align}
where ${s}\sim\mathcal{CN}\left(0,1\right)$\footnote{This assumption is widely adopted in the studies on physical layer security, which facilitates the calculation of the secrecy rate\cite{FullDuplexBaseStation, FullDuplexBaseStation2, RFullDuplex, Huang_TSP}.} is the confidential information-bearing signal, $\mathbf{v}\in\mathbb{C}^{N_T\times 1}$ denotes the beamforming vector, and $\mathbf{n}_a\sim\mathcal{CN}\left(\mathbf{0},\mathbf{\Omega}\right)$ denotes the AN.  The total available power of the BS is denoted by $P_{\mathrm{tot}}$, while the design of $\mathbf{v}$ and $\mathbf{n}_a$ must satisfy the following power constraint:
\begin{align}
\mathrm{tr}\left(\mathbf{v}\mathbf{v}^H+\mathbf{\Omega}\right)\leq P_{\mathrm{tot}}.
\end{align}

Let us assume $z\sim\mathcal{CN}(0,1)$ is the confidential signal transmitted from the UT and its transmit power is $P_t$. In a scheduled slot, the signals received at the BS, the DR and the Eve can be expressed as
\begin{align}
&y_B=\mathbf{r}^H\left(\sqrt{P_t}\mathbf{g}_Bz+\mathbf{H}_L\mathbf{x}_b+\mathbf{n}_B\right),
\label{yB}\\
&{y}_D=\mathbf{h}^H_B\mathbf{x}_b+\sqrt{P_t}gz+{n}_D,
\label{yD}\\
&\mathbf{y}_E=\mathbf{H}_E\mathbf{x}_b+\sqrt{P_t}\mathbf{g}_Ez+\mathbf{n}_E,
\label{yE}
\end{align}
respectively, where $\mathbf{n}_B$, $n_D$, $\mathbf{n}_E$ denote the received noise having a zero mean and unit variance/identity covariance matrix. Since the perfect state information of $\mathbf{H}_L$ is unavailable to the BS, the maximum ratio combining (MRC) receiver, instead of the minimum mean-square error (MMSE) receiver, is assumed to be adopted at the BS for maximizing the signal-to-noise ratio (SNR) of the received confidential signal. The MRC receiver adopted is defined by $\mathbf{r}=\frac{\mathbf{g}_B}{||\mathbf{g}_B||_F}$, hence we have $||\mathbf{r}||_F = 1$.

\subsection{Channel State Information Model}
We consider the quasi-stationary flat-fading channels and assume that the UL/DL pair of channels exhibits reciprocity. At the beginning of the scheduled slot, the BS obtains the CSI of all legitimate channels via training. In practice, the DR broadcasts training sequences to facilitate the UL channel estimation, thus $\mathbf{h}_B^H$ may be estimated at the BS and the inter-terminal interference channel $g$ may be estimated at the UT. Then, the UT can embed $g$ into its own training sequence and broadcast the composite sequence. As a result, both $g$ and $\mathbf{g}_B$ can be acquired at the BS. By contrast, the residual self-interference channel $\mathbf{H}_L$ after adopting digital/analog domain cancellation techniques remains unknown \cite{InbandFullduplex}. Additionally, due to the lack of explicit cooperation between the legitimate network nodes and the Eve, perfect knowledge of the wiretap channels is difficult to obtain at the legitimate network nodes. Hence, we assume that the BS can only acquire imperfect estimates of $\mathbf{H}_E$ and $\mathbf{g}_E$, and we have
\begin{align}
&\mathbf{H}_E=\bar{\mathbf{H}}_E+\Delta\mathbf{H}_E,\quad\mathbf{g}_E=\bar{\mathbf{g}}_E+\Delta \mathbf{g}_E,
\end{align}
where $\bar{\mathbf{H}}_E$ and $\bar{\mathbf{g}}_E$ denote the estimates of $\mathbf{H}_E$ and $\mathbf{g}_E$, respectively, while $\Delta\mathbf{H}_E$ and $\Delta \mathbf{g}_E$ denote the corresponding estimation error \cite{channel_estimate_error1, channel_estimate_error2, channel_estimate_error3, channel_estimate_error4, channel_estimate_error5}. This model is sufficiently generic for characterizing different types of  wiretap channels. More specifically, when the Eve is constituted by an active subscriber, the perfect CSI of the wiretap channel will be available to the associated network nodes, indicating $\Delta\mathbf{H}_E=\mathbf{0}$ and $\Delta \mathbf{g}_E=\mathbf{0}$. On the other hand, when the Eve is completely passive, the BS and UT may know nothing about the CSI of the wiretap channels, implying that $\bar{\mathbf{H}}_E$ and $\bar{\mathbf{g}}_E$ can take arbitrary values.
%In addition to above, such model corresponds to the general case where only the partial CSI of the Eve is available at the BS.

We characterize the uncertainty associated with the CSI estimation and residual self-interference channel by the widely adopted
\emph{deterministic uncertainty model} \cite{RFullDuplex,Huang_TSP}. In this model, $\mathbf{H}_L$, $\Delta\mathbf{H}_E$, and $\Delta \mathbf{g}_E$ are bounded by the sets
$\xi_{\mathbf{H}_L}\triangleq\left\{\mathbf{H}_L: ||\mathbf{H}_L||_F\leq \delta_{\mathbf{H}_L}\right\}$, $\xi_{\mathbf{H}_E}\triangleq\left\{\Delta\mathbf{H}_E: ||\Delta\mathbf{H}_E||_F\leq \delta_{\mathbf{H}_E}\right\}$ and
$\xi_{\mathbf{g}_E}\triangleq\left\{\Delta \mathbf{g}_E: ||\Delta \mathbf{g}_E||_F\leq \delta_{\mathbf{g}_E}\right\}$, respectively, where $\delta_{\mathbf{H}_L}$, $\delta_{\mathbf{H}_E}$ and $\delta_{\mathbf{g}_E}$ are some given constants. Relying on this deterministic model, we optimize the system's secrecy performance under the worst-case channel conditions. This approach guarantees the absolute robustness of the system design, since the achievable secrecy performance would be no worse than that of the worst case.

%\emph{2) Gaussian Uncertainty Model}. In this model, just as \cite{Wang_TSP,Ma_TWC}, the CSI uncertainties $\mathbf{H}_L$, $\Delta\mathbf{H}_E$, and $\Delta \mathbf{c}_E$ follow the Gaussian distribution, i.e.,
%$\mathrm{vec}\left(\mathbf{H}_L\right)\sim\mathcal{CN}\left(\mathbf{0},\xi^2_{\mathbf{H}_L}\mathbf{I}_{N_T\times N_R}\right)$,
%$\Delta\mathbf{c}_E\sim\mathcal{CN}\left(\mathbf{0},\xi^2_{\mathbf{c}_E}\mathbf{I}_{N_E}\right)$, $\mathrm{vec}\left(\Delta\mathbf{H}_E\right)\sim\mathcal{CN}\left(\mathbf{0},\xi^2_{\mathbf{H}_E}\mathbf{I}_{N_T\times N_E}\right)$

\subsection{Problem Formulation}
We seek to optimize the beamformer $\mathbf{v}$ and the covariance matrix $\mathbf{\Omega}$ of AN in order to maximize the worst-case sum secrecy rate. Defining $\mathbf{Q}=\mathbf{v}\mathbf{v}^H$, we can obtain
\begin{align}
&\max_{\mathbf{v},\mathbf{\Omega}\succeq\mathbf{0}}\min_{\substack{\Delta\mathbf{H}_E\in\xi_{\mathbf{H}_E},
\mathbf{H}_L\in\xi_{\mathbf{H}_L},
\Delta\mathbf{g}_E\in\xi_{\mathbf{g}_E}
}}\mathrm{ln}\left(1+\eta_1\right)
+\mathrm{ln}\left(1+\eta_2\right)
\nonumber\\
&\qquad\qquad\qquad\qquad\qquad\qquad\qquad-\ln\mathrm{det}\left(\mathbf{I}_{N_E}+\mathbf{Z}\mathbf{N}^{-1}\right)
\nonumber\\
&\mathrm{s.t.}\quad\quad\quad\mathrm{rank}\left(\mathbf{Q}\right)=1,
\nonumber\\
&\quad \quad \quad \quad \; \mathrm{tr}\left(\mathbf{Q}+\mathbf{\Omega}\right)\leq P_{\mathrm{tot}}, 
\label{worstcaseDesign}
\end{align}
where $\eta_1\triangleq \frac{\mathbf{h}_B^H\mathbf{Q}\mathbf{h}_B}{P_t|g|^2+\mathbf{h}_B^H\mathbf{\Omega}\mathbf{h}_B+1}$,
$\eta_2\triangleq \frac{P_t||\mathbf{g}_B||_F^2}{\mathbf{r}^H\mathbf{H}_L\left(\mathbf{Q}+\mathbf{\Omega}\right)\mathbf{H}_L^H\mathbf{r}+1}$,
$\mathbf{Z}\triangleq \mathbf{H}_E\mathbf{Q}\mathbf{H}_E^H+P_t\mathbf{g}_E\mathbf{g}_E^H$,
$\mathbf{N}\triangleq\mathbf{H}_E\mathbf{\Omega}\mathbf{H}_E^H+\mathbf{I}_{N_E}$.

\section{Robust Joint Beamforming and AN Design: An Iterative Approach}
The max-min problem (\ref{worstcaseDesign}) is non-convex.  In the following, let us transform (\ref{worstcaseDesign}) into a sequence of convex programming problems.

Employing \cite[Proposition 1]{RFullDuplex} and \cite[Lemma 4.1]{InterativeAlgorithm}, we obtain the following equivalent relation:
\begin{align}
&-\ln\left(P_t|g|^2+\mathbf{h}_B^H\mathbf{\Omega}\mathbf{h}_B+1\right)=\nonumber\\
&\max_{a_1>0}\left[-a_1\left(P_t|g|^2+\mathbf{h}_B^H\mathbf{\Omega}\mathbf{h}_B+1\right)+\ln a_1+1\right],
\label{equivalentequation1}\\
&-\ln\left[1+\mathbf{r}^H\mathbf{H}_L\left(\mathbf{v}\mathbf{v}^H+\mathbf{\Omega}\right)\mathbf{H}_L^H\mathbf{r}\right]=\nonumber\\
&\max_{a_2>0}\left\{-a_2\left[1+\mathbf{r}^H\mathbf{H}_L\left(\mathbf{v}\mathbf{v}^H+\mathbf{\Omega}\right)\mathbf{H}_L^H\mathbf{r}\right]+\ln a_2+1\right\},
\label{equivalentequation2}\\
&-\ln \det\left[\mathbf{I}_{N_E}+\mathbf{H}_E\left(\mathbf{v}\mathbf{v}^H+\mathbf{\Omega}\right)\mathbf{H}_E^H+P_t\mathbf{g}_E\mathbf{g}_E^H\right]=\nonumber\\
&N_E+\max_{\mathbf{W}_E\succeq\mathbf{0}}\left[\ln\det\left(\mathbf{W}_E\right)-\mathrm{tr}\left(\mathbf{W}_E\mathbf{X}\right)\right],\label{equivalentequation3}
\end{align}
where $\mathbf{X}\triangleq\mathbf{I}_{N_E}+\mathbf{H}_E\left(\mathbf{v}\mathbf{v}^H+\mathbf{\Omega}\right)\mathbf{H}_E^H+P_t\mathbf{g}_E\mathbf{g}_E^H$. Then, expanding the objective function of (\ref{worstcaseDesign})  and invoking the above relationship, we have
\begin{subequations}\label{ReformulatedProblem1}
\begin{align}
 \max_{\substack{a_1>0,a_2>0,\mathbf{W}_E\succeq\mathbf{0}\\
\mathbf{Q}\succeq\mathbf{0},\mathbf{\Omega}\succeq\mathbf{0}}}& \Xi\left(a_1,a_2,\mathbf{W}_E,\mathbf{Q},\mathbf{\Omega}\right)
\\
\mathrm{s.t.}\quad\quad \; & \mathrm{rank}\left(\mathbf{Q}\right)=1,
\\
& \mathrm{tr}\left(\mathbf{Q}+\mathbf{\Omega}\right)\leq P_{\mathrm{tot}}, 
 \end{align}
\end{subequations}
where the objective function (\ref{ReformulatedProblem1}a) is given in the equation (\ref{Ffunction12}).
\begin{figure*}[!t]
\begin{align}
&\Xi\left(a_1,a_2,\mathbf{W}_E,\mathbf{Q},\mathbf{\Omega}\right)\triangleq
\mathrm{ln}\left(\mathbf{h}_B^H\mathbf{v}\mathbf{v}^H\mathbf{h}_B+P_t|g|^2+\mathbf{h}_B^H\mathbf{\Omega}\mathbf{h}_B+1\right)
 -a_1\left(P_t|g|^2+\mathbf{h}_B^H\mathbf{\Omega}\mathbf{h}_B+1\right)+\ln a_1+1+\nonumber\\
 &\mathrm{ln}\left({P_t||\mathbf{g}_B||_F^2+\mathbf{r}^H\mathbf{H}_L\left(\mathbf{v}\mathbf{v}^H+\mathbf{\Omega}\right)\mathbf{H}_L^H\mathbf{r}+1}\right)
 -a_2\left[1+\mathbf{r}^H\mathbf{H}_L\left(\mathbf{v}\mathbf{v}^H+\mathbf{\Omega}\right)\mathbf{H}_L^H\mathbf{r}\right]+\ln a_2+1-N_E+\ln\det\left(\mathbf{W}_E\right)\nonumber\\
&-\mathrm{tr}\left(\mathbf{W}_E\mathbf{X}\right)
+\mathrm{lndet}\left(\mathbf{I}_{N_E}+\mathbf{H}_E\mathbf{\Omega}\mathbf{H}_E^H\right)
\label{Ffunction12}
\end{align}
\hrulefill
\end{figure*}
\newcounter{TempEqCnt} % ´´½¨ÁÙʱ±äÁ¿TempEqCnt
\setcounter{TempEqCnt}{\value{equation}} % ½«µ±Ç°¹«Ê½ÐòºÅ ¸³¸øTempEqCnt
\setcounter{equation}{16} % µ±Ç°¹«Ê½ÐòºÅ±äΪx£¬xµÈÓÚ³¤¹«Ê½Ó¦ÓеÄÐòºÅ¼Ó1.
\begin{figure*}[!t]
\begin{align}
&F\left(a_1,a_2,\alpha,\beta,\gamma,\mathbf{W}_E,\mathbf{Q},\mathbf{\Omega}\right)\triangleq\mathrm{ln}\left(1+P_t|g|^2+\mathbf{h}_B^H\left(\mathbf{Q}+\mathbf{\Omega}\right)\mathbf{h}_B\right)-a_1\left(P_t|g|^2+\mathbf{h}_B^H\mathbf{\Omega}\mathbf{h}_B+1\right)+\ln a_1+1
\nonumber\\
&-a_2\left(1+\alpha\right)+\ln a_2 +1
+\mathrm{ln}\left(1+P_t||\mathbf{g}_B||_F^2+\alpha\right)
+\ln\det\left(\mathbf{W}_E\right)-\mathrm{tr}\left(\mathbf{W}_E\right)-\beta-P_t\gamma
+\ln \det\left(\mathbf{I}_{N_E}+\mathbf{M}\right)\label{Ffunction}
\end{align}
\hrulefill
\end{figure*}
\setcounter{equation}{\value{TempEqCnt}}
\newcounter{TempEqNo} % ´´½¨ÁÙʱ±äÁ¿TempEqCnt
\setcounter{TempEqNo}{\value{equation}} % ½«µ±Ç°¹«Ê½ÐòºÅ ¸³¸øTempEqCnt
\setcounter{equation}{18} % µ±Ç°¹«Ê½ÐòºÅ±äΪx£¬xµÈÓÚ³¤¹«Ê½Ó¦ÓеÄÐòºÅ¼Ó1.
\begin{figure*}[!t]
\begin{align}
\lambda_{\beta}\left[\begin{matrix}
\mathbf{I}_{N_T\times N_E}&\mathbf{0}\\
\mathbf{0}&-\delta^2_{\mathbf{H}_E}
\end{matrix}
\right]-
\left[\begin{matrix}
\left(\mathbf{Q}+\mathbf{\Omega}\right)^T\otimes\mathbf{W}_E&\left(\left(\mathbf{Q}+\mathbf{\Omega}\right)^T\otimes\mathbf{W}_E\right)\mathrm{vec}\left(\bar{\mathbf{H}}_E\right)\\
\mathrm{vec}\left(\bar{\mathbf{H}}_E\right)^H\left(\left(\mathbf{Q}+\mathbf{\Omega}\right)^T\otimes\mathbf{W}_E\right)&
\mathrm{vec}\left(\bar{\mathbf{H}}_E\right)^H\left(\left(\mathbf{Q}+\mathbf{\Omega}\right)^T\otimes\mathbf{W}_E\right)\mathrm{vec}\left(\bar{\mathbf{H}}_E\right)-\beta
\end{matrix}
\right]\succeq\mathbf{0},\exists \lambda_{\beta}\geq 0.\label{equivalentconstraint2}
\end{align}
\hrulefill
\end{figure*}
\setcounter{equation}{\value{TempEqNo}}

By introducing optimization variables $\alpha,\beta,\gamma$, the problem (\ref{ReformulatedProblem1}) can be reformulated as
\begin{subequations}
\begin{align}
&\max_{\substack{a_1>0,a_2>0,\mathbf{W}_E\succeq\mathbf{0}\\
\alpha>0,\beta>0,\mathbf{Q}\succeq\mathbf{0},\mathbf{\Omega}\succeq\mathbf{0}}}F\left(a_1,a_2,\alpha,\beta,\gamma,\mathbf{W}_E,\mathbf{Q},\mathbf{\Omega}\right)
\nonumber\\
\mathrm{s.t.} \quad &\alpha \geq \mathbf{r}^H\mathbf{H}_L\left(\mathbf{Q}+\mathbf{\Omega}\right)\mathbf{H}_L^H\mathbf{r},\quad  \forall\mathbf{H}_L\in\xi_{\mathbf{H}_L};\label{Newconstraint1}\\
&\beta\geq \mathrm{tr}\left(\mathbf{W}_E\left(\mathbf{H}_E\left(\mathbf{Q}+\mathbf{\Omega}\right)\mathbf{H}_E^H\right)\right),\; \forall \Delta\mathbf{H}_E\in\xi_{\Delta\mathbf{H}_E};\label{Newconstraint2}\\
&\gamma\geq \mathrm{tr}\left(\mathbf{W}_E\left(P_t\mathbf{g}_E\mathbf{g}_E^H\right)\right),\quad\forall \Delta\mathbf{g}_E\in\xi_{\Delta\mathbf{g}_E};\label{Newconstraint3}\\
&\mathbf{M}\preceq \mathbf{H}_E\mathbf{\Omega}\mathbf{H}_E^H \quad \forall \Delta\mathbf{H}_E\in\xi_{\Delta\mathbf{H}_E};\label{Newconstraint4}\\
&\mathrm{rank}\left(\mathbf{Q}\right)=1;\label{Newconstraint5}
\\
& \mathrm{tr}\left(\mathbf{Q}+\mathbf{\Omega}\right)\leq P_{\mathrm{tot}},\label{NewPowerConstraint}
\end{align}
\end{subequations}
where the objective function $F\left(a_1,a_2,\alpha,\beta,\gamma,\mathbf{W}_E,\mathbf{Q},\mathbf{\Omega}\right)$ is defined in  (\ref{Ffunction}).

Let us first transform the constraint (\ref{Newconstraint1}) into a more convenient formulation:
\begin{align}
&\mathbf{r}^H\mathbf{H}_L\left(\mathbf{Q}+\mathbf{\Omega}\right)\mathbf{H}_L^H\mathbf{r}
=\mathrm{tr}\left(\mathbf{r}\mathbf{r}^H\mathbf{H}_L\left(\mathbf{Q}+\mathbf{\Omega}\right)\mathbf{H}_L^H\right)
\nonumber\\
&\overset{(a)}{=}
\left(\mathrm{vec}\left(\mathbf{H}_L\right)\right)^H\left(\left(\mathbf{Q}+\mathbf{\Omega}\right)^T\otimes\mathbf{r}\mathbf{r}^H\right)\mathrm{vec}\left(\mathbf{H}_L\right),
\label{originalProblem}
\end{align}
where the equation $(a)$ is due to the following equation:
\begin{align}
\mathrm{tr}\left(\mathbf{ABCD}\right)=\left(\mathrm{vec}\left(\mathbf{D}^T\right)\right)^{T}\left(\mathbf{C}^T\otimes\mathbf{A}\right)\mathrm{vec}\left(\mathbf{B}\right).
\end{align}
%\begin{figure*}[!t]
%\begin{align}
%&\Xi\left(a_1,a_2,\mathbf{W}_E,\mathbf{Q},\mathbf{\Omega}\right)\triangleq
%\mathrm{ln}\left(\mathbf{h}_B^H\mathbf{v}\mathbf{v}^H\mathbf{h}_B+P_t|g|^2+\mathbf{h}_B^H\mathbf{\Omega}\mathbf{h}_B+1\right)
% -a_1\left(P_t|g|^2+\mathbf{h}_B^H\mathbf{\Omega}\mathbf{h}_B+1\right)+\ln a_1+1+\nonumber\\
% &\mathrm{ln}\left({P_t||\mathbf{g}_B||_F^2+\mathbf{r}^H\mathbf{H}_L\left(\mathbf{v}\mathbf{v}^H+\mathbf{\Omega}\right)\mathbf{H}_L^H\mathbf{r}+1}\right)
% -a_2\left[1+\mathbf{r}^H\mathbf{H}_L\left(\mathbf{v}\mathbf{v}^H+\mathbf{\Omega}\right)\mathbf{H}_L^H\mathbf{r}\right]+\ln a_2+1-N_E+\ln\det\left(\mathbf{W}_E\right)\nonumber\\
%&-\mathrm{tr}\left(\mathbf{W}_E\mathbf{X}\right)
%+\mathrm{lndet}\left(\mathbf{I}_{N_E}+\mathbf{H}_E\mathbf{\Omega}\mathbf{H}_E^H\right)
%\label{Ffunction15}
%\end{align}
%\hrulefill
%\end{figure*}
With similar procedures, the constraint (\ref{Newconstraint2}) -  (\ref{Newconstraint3}) can be transformed into a more convenient formulation, and
the non-convex problem (\ref{worstcaseDesign}) can be equivalently reformulated as:
\begin{subequations}\label{RobustProblem}
\begin{align}
&\max_{\substack{a_1>0,a_2>0,\mathbf{W}_E\succeq\mathbf{0}\\
\alpha>0,\beta>0,\mathbf{Q}\succeq\mathbf{0},\mathbf{\Omega}\succeq\mathbf{0}}}F\left(a_1,a_2,\alpha,\beta,\gamma,\mathbf{W}_E,\mathbf{Q},\mathbf{\Omega}\right)
\nonumber\\
&\mathrm{s.t.}\quad\alpha\geq \left(\mathrm{vec}\left(\mathbf{H}_L\right)\right)^H\left(\left(\mathbf{Q}+\mathbf{\Omega}\right)^T\otimes\mathbf{r}\mathbf{r}^H\right)\mathrm{vec}\left(\mathbf{H}_L\right)
,\nonumber\\
&\qquad\forall \mathbf{H}_L\in\xi_{\mathbf{H}_L};\label{constraint1}\\
&\beta\geq \left(\mathrm{vec}\left(\Delta\mathbf{H}_E\right)\right)^H\left(\left(\mathbf{Q}+\mathbf{\Omega}\right)^T\otimes\mathbf{W}_E^H\right)\mathrm{vec}\left(\Delta\mathbf{H}_E\right)
+\nonumber\\
&2\mathrm{Re}\left(\left(\mathrm{vec}\left(\Delta\mathbf{H}_E\right)\right)^H\left(\left(\mathbf{Q}+\mathbf{\Omega}\right)^T\otimes\mathbf{W}_E^H\right)\mathrm{vec}\left(\bar{\mathbf{H}}_E\right)
\right)+
\nonumber\\
&\mathrm{vec}\left(\bar{\mathbf{H}}_E\right)^H\left(\left(\mathbf{Q}+\mathbf{\Omega}\right)^T\otimes\mathbf{W}_E^H\right)\mathrm{vec}\left(\bar{\mathbf{H}}_E\right)
,\nonumber\\
&\qquad\forall \Delta\mathbf{H}_E\in\xi_{\Delta\mathbf{H}_E};\label{constraint2}\\
&\gamma\geq \Delta\mathbf{g}_E^H\mathbf{W}_E\Delta\mathbf{g}_E+2\mathrm{Re}\left( \Delta\mathbf{g}_E^H\mathbf{W}_E\bar{\mathbf{g}}_E\right)+
\bar{\mathbf{g}}_E^H\mathbf{W}_E\bar{\mathbf{g}}_E,\nonumber\\
&\qquad\forall \Delta\mathbf{g}_E\in\xi_{\Delta\mathbf{g}_E};\label{constraint3}\\
&\mathbf{M}\preceq \Delta\mathbf{H}_E\mathbf{\Omega}\Delta\mathbf{H}_E^H
+\Delta\mathbf{H}_E\mathbf{\Omega}\bar{\mathbf{H}}_E^H+
\bar{\mathbf{H}}_E\mathbf{\Omega}\Delta\mathbf{H}_E^H
\nonumber\\
&\qquad+\bar{\mathbf{H}}_E\mathbf{\Omega}\bar{\mathbf{H}}_E^H ,\forall \Delta\mathbf{H}_E\in\xi_{\Delta\mathbf{H}_E};
\label{constraint4}\\
&\mathrm{rank}\left(\mathbf{Q}\right)=1;\label{constraint5}
\\
& \mathrm{tr}\left(\mathbf{Q}+\mathbf{\Omega}\right)\leq P_{\mathrm{tot}}.\label{PowerConstraint}
\end{align}
\end{subequations}
%\begin{figure*}[!t]
%\begin{align}
%&F\left(a_1,a_2,\alpha,\beta,\gamma,\mathbf{W}_E,\mathbf{Q},\mathbf{\Omega}\right)\triangleq\mathrm{ln}\left(1+P_t|g|^2+\mathbf{h}_B^H\left(\mathbf{Q}+\mathbf{\Omega}\right)\mathbf{h}_B\right)-a_1\left(P_t|g|^2+\mathbf{h}_B^H\mathbf{\Omega}\mathbf{h}_B+1\right)+\ln a_1+1
%\nonumber\\
%&-a_2\left(1+\alpha\right)+\ln a_2 +1
%+\mathrm{ln}\left(1+P_t||\mathbf{g}_B||_F^2+\alpha\right)
%+\ln\det\left(\mathbf{W}_E\right)-\mathrm{tr}\left(\mathbf{W}_E\right)-\beta-P_t\gamma
%+\ln \det\left(\mathbf{I}_{N_E}+\mathbf{M}\right)\label{Ffunction}
%\end{align}
%\hrulefill
%\end{figure*}

The resultant problem (\ref{RobustProblem}) still remains challenging to solve. To obtain a more tractable problem, the following reformulation is carried out.

According to the eigenvalue equation, the constraint (\ref{constraint1}) can be rewritten as
\setcounter{equation}{17}
\begin{align}
\alpha \geq \delta^2_{\mathbf{H}_L}\lambda_{\max}\left(\left(\mathbf{Q}+\mathbf{\Omega}\right)^T\otimes\mathbf{rr}^H\right),
\label{equivalentconstraint1}
\end{align}
where $\lambda_{\max}$ is the maximal eigenvalue of  $\left(\mathbf{Q}+\mathbf{\Omega}\right)^T\otimes\mathbf{rr}^H$.

According to the S-Procedure \cite{ConvexOptimization}, the constraints (\ref{constraint2}) and (\ref{constraint3}) hold if and only if there exist $\lambda_{\beta}\geq 0$ and $\lambda_{\gamma}\geq 0$ such that (\ref{equivalentconstraint2}) and (\ref{equivalentconstraint3}) hold: 
\setcounter{equation}{19}
%\begin{figure*}[!t]
%\begin{align}
%\lambda_{\beta}\left[\begin{matrix}
%\mathbf{I}_{N_T\times N_E}&\mathbf{0}\\
%\mathbf{0}&-\delta^2_{\mathbf{H}_E}
%\end{matrix}
%\right]-
%\left[\begin{matrix}
%\left(\mathbf{Q}+\mathbf{\Omega}\right)^T\otimes\mathbf{W}_E&\left(\left(\mathbf{Q}+\mathbf{\Omega}\right)^T\otimes\mathbf{W}_E\right)\mathrm{vec}\left(\bar{\mathbf{H}}_E\right)\\
%\mathrm{vec}\left(\bar{\mathbf{H}}_E\right)^H\left(\left(\mathbf{Q}+\mathbf{\Omega}\right)^T\otimes\mathbf{W}_E\right)&
%\mathrm{vec}\left(\bar{\mathbf{H}}_E\right)^H\left(\left(\mathbf{Q}+\mathbf{\Omega}\right)^T\otimes\mathbf{W}_E\right)\mathrm{vec}\left(\bar{\mathbf{H}}_E\right)-\beta
%\end{matrix}
%\right]\succeq\mathbf{0},\exists \lambda_{\beta}\geq 0.\label{equivalentconstraint2}
%\end{align}
%\hrulefill
%\end{figure*}
\begin{align}
&\lambda_{\gamma}\left[\begin{matrix}
\mathbf{I}_{N_T\times N_E}&\mathbf{0}\\
\mathbf{0}&-\delta^2_{\mathbf{g}_E}
\end{matrix}
\right]-
\left[\begin{matrix}
\mathbf{W}_E&\mathbf{W}_E\bar{\mathbf{g}}_E\\
\bar{\mathbf{g}}_E^H\mathbf{W}_E&
\bar{\mathbf{g}}_E^H\mathbf{W}_E\bar{\mathbf{g}}_E-\gamma
\end{matrix}
\right]\succeq\mathbf{0},
\nonumber\\
&\exists \lambda_{\gamma}\geq 0.\label{equivalentconstraint3}
\end{align}

Then, invoking the robust quadratic matrix
inequality of \cite{Wang_TVT}, the constraint (\ref{constraint4}) holds if and only if there exists a value $\lambda_{\mathbf{M}} \geq 0$ such that
\begin{align}
&\left[\begin{matrix}
\bar{\mathbf{H}}_E\mathbf{\Omega}\bar{\mathbf{H}}_E^H-\mathbf{M}&\bar{\mathbf{H}}_E\mathbf{\Omega}\\
\mathbf{\Omega}\bar{\mathbf{H}}^H_E&
\mathbf{\Omega}
\end{matrix}
\right]
-\lambda_{\mathbf{M}}\left[\begin{matrix}
\mathbf{I}_{N_E}&\mathbf{0}\\
\mathbf{0}&-\frac{1}{\delta^2_{\mathbf{H}_E}}
\end{matrix}
\right]\succeq\mathbf{0},
\nonumber\\
&\exists \lambda_{\mathbf{M}}\geq 0.\label{equivalentconstraint4}
\end{align}

Finally, invoking the classic semidefinite relaxation (SDR) technique\cite{Yang_SDR, Yang_JSAC}, we drop the constraint (\ref{constraint5}) to obtain the rank-relaxation version of  (\ref{RobustProblem}).

As a result, the robust joint design of the confidential signal's beamforming vector $\mathbf{v}$ and  the covariance matrix $\mathbf{\Omega}$ of the AN can be reformulated as
\begin{align}
&\max_{\substack{a_1>0,a_2>0,\mathbf{W}_E\succeq\mathbf{0}\\
\alpha>0,\beta>0,\mathbf{Q}\succeq\mathbf{0},\mathbf{\Omega}\succeq\mathbf{0}}}F\left(a_1,a_2,\alpha,\beta,\gamma,\mathbf{W}_E,\mathbf{Q},\mathbf{\Omega}\right)
\nonumber\\
&\mathrm{s.t.}\qquad(\ref{equivalentconstraint1}),(\ref{equivalentconstraint2}),(\ref{equivalentconstraint3}),(\ref{equivalentconstraint4})\, \mathrm{and}\,(\textrm{\ref{PowerConstraint}}).\label{Alternatingoptimization}
\end{align}
Now, although the problem (\ref{Alternatingoptimization}) itself still remains non-convex and its global optimum is difficult to obtain, it can be decomposed into two sub-problems that are convex. Then, we can adopt the efficient
BCD algorithm of \cite[Section 2.7]{NolinearProgramming} to find a sequence of locally optimal solutions of (\ref{Alternatingoptimization}). More specifically, in this method, all optimization variables are decoupled into two blocks:
$\left\{\alpha,\beta,\mathbf{Q},\mathbf{\Omega}\right\}$ and
$\left\{a_1,a_2,\mathbf{W}_E\right\}$. Note that when the block $\left\{a_1,a_2,\mathbf{W}_E\right\}$ is fixed, the problem (\ref{Alternatingoptimization}) becomes convex with respect to $\left\{\alpha,\beta,\mathbf{Q},\mathbf{\Omega}\right\}$, and  when the block $\left\{\alpha,\beta,\mathbf{Q},\mathbf{\Omega}\right\}$ is fixed, the problem (\ref{Alternatingoptimization}) becomes convex with respect to $\left\{a_1,a_2,\mathbf{W}_E\right\}$.
Hence, the two blocks of variables can be optimized in turn with low complexity by fixing one block and optimizing the other.\footnote{As pointed out in \cite{real_time_convex_optimization}, the combination of
substantially increased computing power, sophisticated algorithms,
and new coding approaches has made it possible to
solve modest-sized convex optimization problems on
microsecond or millisecond time scales and with strict
completion deadlines. This enables real-time convex optimization
in signal processing. In fact, this time-scale is also similar to a single transmission time interval (TTI) in 4G LTE (one millisecond) and in 5G NR (a fraction of one millisecond).}

The BCD algorithm employed is summarized in Algorithm 1.
\begin{algorithm}[!t]
\small
\small{\caption{\small The BCD algorithm for solving the robust joint optimization problem (\ref{Alternatingoptimization}).}} % �㷨�ı���
\label{alg_1}%���㷨һ����ǩ���������������ж��㷨������
\begin{algorithmic}%��֪[1]�Ǹ����ģ�
\STATE
Set $l=1$ and perform initialization with arbitrary feasible $\bar{a}_1(l),\bar{a}_2(l),\bar{\mathbf{W}}_E(l)$ .

\STATE
\WHILE{the difference between the values of $F\left(\bar{a}_1(l),\bar{a}_2(l),\bar{\alpha}(l),\bar{\beta}(l),\bar{\gamma}(l),\bar{\mathbf{W}}_E(l),\bar{\mathbf{Q}}(l),\bar{\mathbf{\Omega}}(l)\right)
$ in successive iterations is larger than $\phi$ for some $\phi >0$,}
\STATE 1. Solve the problem (\ref{Alternatingoptimization}) by fixing ${a}_1 = \bar{a}_1(l),{a}_2 = \bar{a}_2(l),{\mathbf{W}}_E = \bar{\mathbf{W}}_E(l)$, and obtain the global optimum, i.e., $\alpha^*,\beta^*,\gamma^*,\mathbf{Q}^*,\mathbf{\Omega}^*$.

\STATE 2. Set $\bar{\alpha}(l) = \alpha^*,\bar{\beta}(l)=\beta^*,\bar{\gamma}(l)=\gamma^*$, $\bar{\mathbf{Q}}(l) = \mathbf{Q}^*$, and $\bar{\mathbf{\Omega}}(l) = \mathbf{\Omega}^*$.
\STATE 3. Solve the problem (\ref{Alternatingoptimization}) by fixing $\alpha = \bar{\alpha}(l) ,\beta = \bar{\beta}(l) ,\gamma = \bar{\gamma}(l),\mathbf{Q}=\bar{\mathbf{Q}}(l),\mathbf{\Omega}=\bar{\mathbf{\Omega}}(l)$ and obtain the global optimum, i.e., $a_1^*,a_2^*,{\mathbf{W}}_E^*$.
\STATE 4. $l=l+1$,
\STATE 5. Set $\bar{a}_1(l) = a_1^*,\bar{a}_2(l) = a_2^*,$ and $\bar{\mathbf{W}}_E(l) = \mathbf{W}_E^*$.
\ENDWHILE
\STATE \textbf{Output:} $\mathbf{Q},\mathbf{\Omega}$.
\end{algorithmic}
\end{algorithm}
Since Algorithm 1 yields non-descending objective function values, it must converge subject to the plausible constraint that physically the secrecy rate is finite.

Given Algorithm 1, the optimal $\mathbf{Q}$ can be found by solving (\ref{Alternatingoptimization}). However, there is no guarantee that $\mathrm{rank}\left(\mathbf{Q}\right)=1$ and that $\mathbf{v}$ can be derived from $\mathbf{Q}$ without any performance deterioration. In the following, we propose a simple method of constructing $\mathbf{v}$ from $\mathbf{Q}$ with good guaranteed performance, as summarized in Algorithm 2.
\begin{algorithm}[!t]%�㷨�Ŀ�'
\small
\caption{\small{The proposed BCD algorithm for solving the robust joint optimization problem (\ref{Alternatingoptimization}).}} %�㷨�ı���
\label{alg_2}%���㷨һ����ǩ���������������ж��㷨������
\begin{algorithmic}%��֪[1]�Ǹ����ģ�
\STATE
1. Set $\hat{\mathbf{Q}}=\frac{\mathbf{Qh}_B\mathbf{h}^H_B\mathbf{Q}}{\mathbf{h}_B^H\mathbf{Q}\mathbf{h}_B}$.
\STATE
2. $\mathbf{v}$ is generated from $\hat{\mathbf{Q}}$ by singular value decomposition.
\end{algorithmic}
\end{algorithm}
\begin{pp}
The secrecy performance achieved by $\hat{\mathbf{Q}}$ constructed in Algorithm 2 is not inferior to the one achieved by ${\mathbf{Q}}$.
\end{pp}
\begin{proof}
The objective function value of the problem (\ref{worstcaseDesign}) represents the achievable secrecy rate performance. Then, let us replace $\mathbf{v}\mathbf{v}^H$ by $\hat{\mathbf{Q}}$ and ${\mathbf{Q}}$, respectively, in the problem (\ref{worstcaseDesign}) to give the proof.

Firstly, we have the following linear matrix inquality
\begin{align}
{\mathbf{Q}}-\hat{\mathbf{Q}}
=\mathbf{Q}^{1/2}\left(\mathbf{I}-\frac{\mathbf{Q}^{1/2}\mathbf{h}_B\mathbf{h}^H_B\mathbf{Q}^{1/2}}{\mathbf{h}_B^H\mathbf{Q}\mathbf{h}_B}\right)\mathbf{Q}^{1/2}
\succeq\mathbf{0}.
\end{align}
Thus, $\mathrm{tr}\left({\mathbf{Q}}\right)\geq\mathrm{tr}\left(\hat{\mathbf{Q}}\right)$. Then, we can conclude that if ${\mathbf{Q}}$ satisfies the power constraint of the problem (\ref{worstcaseDesign}), $\hat{\mathbf{Q}}$ should also satisfy this power constraint.

Next, let us check the objective function of the problem (\ref{worstcaseDesign}), where $\mathbf{v}\mathbf{v}^H$ appears in the numerator of $\eta_1$, in the denominator of $\eta_2$, as well as in $\mathbf{Z}$. Then, for establishing the proof, the following inequalities must be proved:
\begin{align}
\quad\mathbf{h}_B^H{\mathbf{Q}}\mathbf{h}_B\leq \mathbf{h}_B^H\hat{\mathbf{Q}}\mathbf{h}_B,
\label{inequality1}\\
\mathbf{r}^H\mathbf{H}_L{\mathbf{Q}}\mathbf{H}_L^H\mathbf{r}\geq \mathbf{r}^H\mathbf{H}_L\hat{\mathbf{Q}}\mathbf{H}_L^H\mathbf{r},
\label{inequality2}\\
\mathbf{H}_E{\mathbf{Q}}\mathbf{H}_E^H\succeq \mathbf{H}_E\hat{\mathbf{Q}}\mathbf{H}_E^H.\label{inequality3}
\end{align}
Observe that (\ref{inequality1}) holds true, since $\mathbf{h}_B^H\hat{\mathbf{Q}}\mathbf{h}_B=\mathbf{h}_B^H{\mathbf{Q}}\mathbf{h}_B$. Additionally, since ${\mathbf{Q}}-\hat{\mathbf{Q}}\succeq\mathbf{0}$, (\ref{inequality2}) and (\ref{inequality3}) also hold true.
\end{proof}

The following proposition shows that the limit points generated by Algorithm 1 would satisfy the KKT condition of the problem (\ref{worstcaseDesign}).
\begin{pp}
The limit point of the sequence $\left\{\bar{a}_1(l),\bar{a}_3(l),\bar{\alpha}(l),\bar{\beta}(l),\bar{\gamma}(l),\bar{\mathbf{W}}_E(l),\bar{\mathbf{Q}}(l),\bar{\mathbf{\Omega}}(l)\right\}$ generated by Algorithm 1, i.e.,
$\left\{\bar{a}_1^*,\bar{a}_3^*,\bar{\alpha}^*,\bar{\beta}^*,\bar{\gamma}^*,\bar{\mathbf{W}}_E^*,\bar{\mathbf{Q}}^*,\bar{\mathbf{\Omega}}^*\right\}$
is a KKT point of the nonconvex optimization problem (\ref{worstcaseDesign}).
\end{pp}
\begin{proof}
The proof is given in Appendix B.
\end{proof}
\section{Simulation Results}
In our simulations, we assume that all the available channel coefficients obey $\mathcal{CN}\left({0},1\right)$. The transmit power of the UT is set to $P_t=20$ dBm and the received noise power is normalized to 0 dBm.
The simulation results were obtained by averaging over 200 independent trials. With the obtained
$\mathbf{v}$ and $\mathbf{\Omega}$,  the achievable worst-case sum secrecy rate $R_{w}$ can be
computed from (\ref{worstcaseratecalculation}) of Appendix A. Note that mathematically it is possible for the sum secrecy rate $R_{w}$ to be negative when the transmit power $P_{tot}$ is low, but in the real world, the achievable secrecy rate cannot be negative, since the information wiretapped by the eavesdropper cannot be higher than that transmitted by the legitimate transmitter. Hence, when we encounter negative secrecy sum rate in our simulations, we have to use $\max(R_{w},0)$ to ensure that its lowest value is zero.
Assuming $N_T=3,N_R=2$ and $N_E=2$,  in Fig. \ref{PowerPower} we show the achievable average $R_{w}$ versus the total available transmit power $P_{\mathrm{tot}}$ of the FD BS, subject to the error bounds $\delta_{\mathbf{H}_E}$, $\delta_{\mathbf{H}_L}$ and $\delta_{\mathbf{g}_E}$. When  $P_{\mathrm{tot}}$ increases, the strength of both the confidential signals and the AN also increases. Therefore, the average $R_{w}$ increases with $P_{\mathrm{tot}}$, which has been validated by our simulation results in Fig.  \ref{PowerPower}. On the other hand, as the error bounds increase,
the CSI estimation of $\mathbf{H}_E$, $\mathbf{H}_L$ and $\mathbf{g}_E$ becomes more and more inaccurate. Therefore, both the information leakage to the Eve  and the jamming signal (i.e., the AN) leakage to the BS increase with the error bounds, which results in a secrecy performance degradation. This argument has been validated by our simulation results in Fig.  \ref{PowerPower}, showing that smaller error bounds result in a larger average $R_w$.
% e.g.,  the average $R_w$ with
%$\delta_{\mathbf{H}_E}=0.03,\delta_{\mathbf{H}_L}= 0.01,\delta_{\mathbf{g}_E}=0.02$ is larger than the one with
%$\delta_{\mathbf{H}_E}=0.06,\delta_{\mathbf{H}_L}= 0.03,\delta_{\mathbf{g}_E}=0.05$.

\begin{figure}[tbp]
\centering
\includegraphics[width=3.5in]{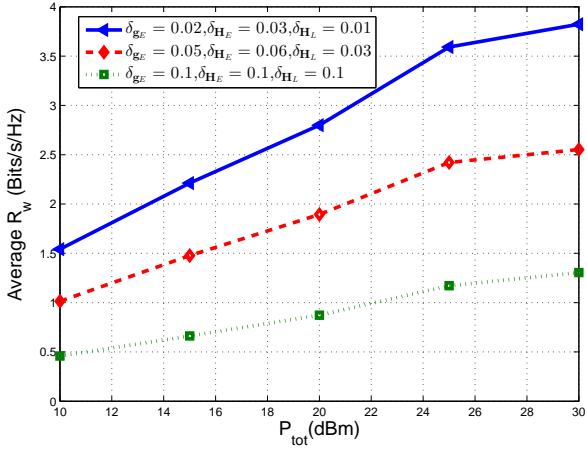}
\caption{The average worst-case sum secrecy rate versus $P_{\mathrm{tot}}$ under $N_T=3,N_R=2,N_E=2$, and different values of error bounds.}
\label{PowerPower}
\end{figure}
\begin{figure}[tbp]
\centering
\includegraphics[width=3.5in]{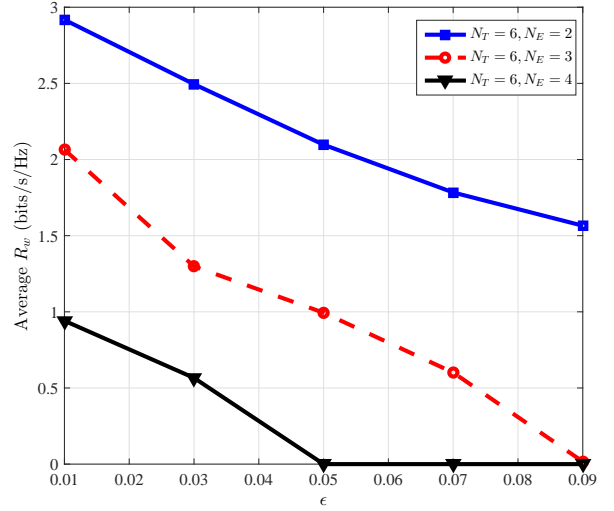}
\caption{The average worst-case sum secrecy rate versus the error bound $\epsilon$ under $N_R=2,P_{\mathrm{tot}}=10$ dBm, and different values of $N_E$.}
\label{Epsilon}
\end{figure}

For illustrating the robustness of the proposed secrecy transmission scheme to imperfect CSI,  we show in Fig. \ref{Epsilon} how the achievable average $R_w$ changes upon increasing the error bounds.  For simplicity, we assume $\delta_{\mathbf{H}_E}=\delta_{\mathbf{H}_L}=\delta_{\mathbf{g}_E}=\epsilon$, $P_{\mathrm{tot}}=10$ dBm and $N_R=2$.
We observe that the achievable average $R_w$ decreases upon increasing $\epsilon$, and the secrecy performance degradation becomes significant when $\epsilon$ increases from 0.01 to 0.09. Additionally, we see that a larger $N_E$ results in a smaller average $R_w$. This is because when $N_E$ becomes larger, the wiretapping capability of the Eve becomes stronger, and the achievable secrecy performance is degraded.

\begin{figure}[tbp]
\centering
\includegraphics[width=3.5in]{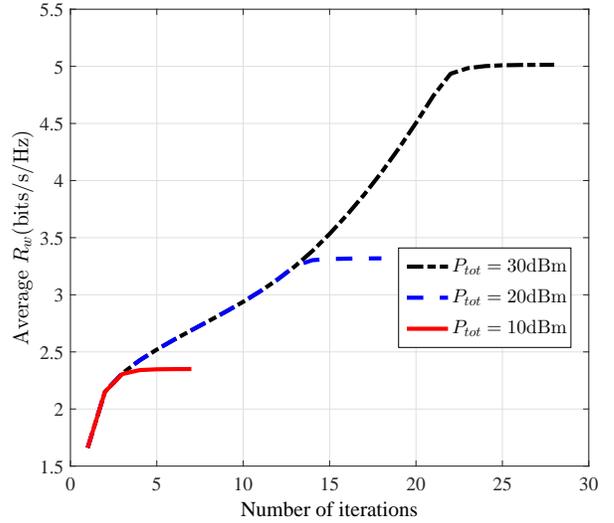}
\caption{The average worst-case sum secrecy rate versus $P_{\mathrm{tot}}$ under $N_T=4,N_R=2,N_E=2,P_t=5$dBm.}
\label{InterationNUmber}
\end{figure}
In Fig. \ref{InterationNUmber} we characterize the convergence behaviour of Algorithm 1 subject to different values of $P_{tot}$. It is observed that the number of iterations required by Algorithm 1 increases when $P_{tot}$ becomes larger. This is because upon  increasing  $P_{tot}$, the feasible set of the problem (\ref{Alternatingoptimization}) is expanded. Additionally, we see that Algorithm 1 indeed converges in all of our observations recorded.

\begin{figure}[tbp]
\centering
\includegraphics[width=3.5in]{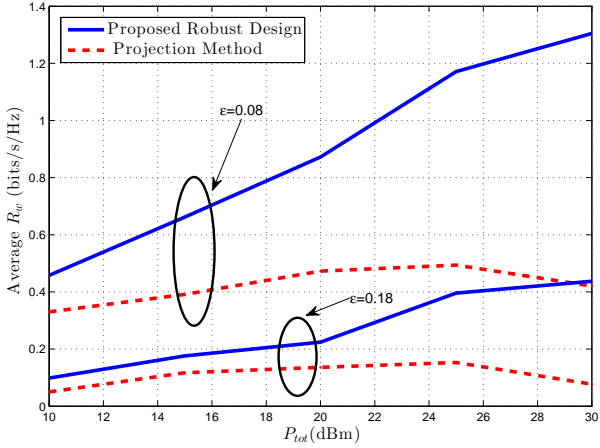}
\caption{Performance comparison between the proposed robust strategy and the projection method under $N_T=3,N_R=2,N_E=2$dBm.}
\label{PerformanceComparison}
\end{figure}
For characterizing the performance of our proposed algorithm, we use projection matrix theory to construct a benchmarker (termed as ''Projection Method''), and use the analysis results derived in Appendix A to quantify its worst-case secrecy rate. Since the CSI of the eavesdropping channel is imperfect, the precoder $\mathbf{v}$ is constructed by  projecting $\mathbf{h}_B$ onto the null-space of the known part of the CSI of the eavesdropping channel, i.e., $\bar{\mathbf{H}}_E$. Explicitly,
\begin{align}
\mathbf{v} = P_t\frac{\mathbf{U}_E\mathbf{U}^H_E\mathbf{h}_B}{||\mathbf{U}_E\mathbf{U}^H_E\mathbf{h}_B||_F},
\end{align}
where, $P_t$ is the transmit power of the information signal, and $\mathbf{U}_E$ is the null-space of $\bar{\mathbf{H}}_E$. For deigning AN, since the self-interference channel $\mathbf{H}_L$ and the eavesdropping channel $\mathbf{H}_E$ are both imperfect known,
we construct the artificial noise in the null-space of $\mathbf{h}_B$ and allocate the power of the artificial noise equally.
Explicitly,
\begin{align}
\mathbf{\Omega} = \frac{P_N\mathbf{U}_{\mathbf{h}_B}\mathbf{U}_{\mathbf{h}_B}^H}{N_T-1},
\end{align}
where, $P_N$ is the total power of the artificial noise and $\mathbf{U}_{\mathbf{h}_B}$ is the null-space of $\mathbf{h}_B$.
We allocate the transmit power equally between the information signal and AN, i.e.,
$P_t=P_N=\frac{P_{tot}}{2}$.

Fig. \ref{PerformanceComparison} show the secrecy performance achieved both by our proposed robust strategy and by the ''Projection Method''. Based on our simulation results, we can find that the secrecy performance achieved by the proposed robust strategy is much  better  than that of the ''Projection Method'', and the gains become higher upon increasing $P_{tot}$, which has confirmed the efficiency of our proposed robust strategy.

\section{Conclusions}
In this paper, we investigated the worst-case physical layer security of a FD system composed of a FD multi-antenna BS, a ``sophisticated/strong'' multi-antenna eavesdropper, a single-antenna UT and a single-antenna DR. Assuming that only imperfect CSI of the eavesdropping and residual self-interference channels is available to the FD BS, we have proposed a robust secrecy transmission scheme, where the AN and beamforming vector are jointly designed for securing the confidential signals transmitted from the BS and UT. By employing the BCD algorithm and linear matrix inequality, we transform the worst-case non-convex problem into a sequence of convex problems in order to find its locally optimal solution. For evaluating the achievable performance, the analysis result of the worst-case secrecy rate is derived. Furthermore, by employing the projection matrix theory, we construct another secrecy transmission scheme for the performance comparison. Our simulation results have validated the effectiveness of the proposed secrecy transmission scheme and provided valuable insights into how several relevant design parameters affect the achievable secrecy performance of the system considered.

\appendices
\section{Achievable Worst-Case Sum Secrecy Rate Calculation}
Invoking Algorithm 1 and Algorithm 2, we obtain the optimized rank-1 $\mathbf{Q}$ and $\mathbf{\Omega}$.
Then, the achievable worst-case sum secrecy rate $R_{w}$ can be calculated from the objective function of the problem (\ref{worstcaseDesign}).
Firstly, the minimum of $\eta_2$, i.e, $\eta_2^{\min}$ can be calculated as $\eta_2^{\min}=\frac{P_t||\mathbf{g}_B||_F^2}{\delta^2_{\mathbf{H}_L}\lambda_{\max}\left(\left(\mathbf{Q}+\mathbf{\Omega}\right)^T\otimes\mathbf{rr}^H\right)
+1}$.

Then, we calculate the maximum of $\ln\det\left(\mathbf{I}_{N_E}+\mathbf{Z}\mathbf{N}^{-1}\right)$, which can be expanded as
\begin{align}
&\ln\det\left(\mathbf{I}_{N_E}+\mathbf{H}_E\mathbf{Q}\mathbf{H}_E^H\left(\mathbf{H}_E\mathbf{\Omega}\mathbf{H}_E^H+\mathbf{I}_{N_E}\right)^{-1}\right)-
\nonumber\\
&\ln\left(1+\mathrm{tr}\left(P_t\mathbf{g}_E\mathbf{g}_E^H\left(\mathbf{H}_E\left(\mathbf{\Omega}+\mathbf{Q}\right)\mathbf{H}_E^H+\mathbf{I}_{N_E}\right)^{-1}\right)\right)\overset{(a)}{=}
\nonumber\\
&\ln\left(1+\mathrm{tr}\left(\mathbf{H}_E\mathbf{Q}\mathbf{H}_E^H\left(\mathbf{H}_E\mathbf{\Omega}\mathbf{H}_E^H+\mathbf{I}_{N_E}\right)^{-1}\right)\right)-
\nonumber\\
&\ln\left(1+\mathrm{tr}\left(P_t\mathbf{g}_E\mathbf{g}_E^H\left(\mathbf{H}_E\left(\mathbf{\Omega}+\mathbf{Q}\right)\mathbf{H}_E^H+\mathbf{I}_{N_E}\right)^{-1}\right)\right)\nonumber.
\end{align}
Step $(a)$ holds true, since $\mathrm{rank}\left(\mathbf{Q}\right)=1$.

Then, the upper bound of $\ln\det\left(\mathbf{I}_{N_E}+\mathbf{Z}\mathbf{N}^{-1}\right)$ is given by $\mathrm{ln}(1+\theta^*_1)+\mathrm{ln}(1+\theta^*_2)$, where
\begin{align}
\theta^*_1=&{\arg\min}_{\theta_1}\theta_1
\nonumber\\
&\mathrm{s.t.}\, \theta_1 \geq \mathrm{tr}\left(\mathbf{H}_E\mathbf{Q}\mathbf{H}_E^H\left(\mathbf{H}_E\mathbf{\Omega}\mathbf{H}_E^H+\mathbf{I}_{N_E}\right)^{-1}\right),
\nonumber\\
&\forall \Delta\mathbf{H}_E\in\xi_{\mathbf{H}_E};\label{theta1}
\\
\theta^*_2=&{\arg\min}_{\theta_2}\theta_2
\nonumber\\
&\mathrm{s.t.}\, \theta_2\geq P_t\mathrm{tr}\left(\mathbf{g}_E\mathbf{g}_E^H\left(\mathbf{H}_E\left(\mathbf{\Omega}+\mathbf{Q}\right)\mathbf{H}_E^H+\mathbf{I}_{N_E}\right)^{-1}\right),
\nonumber\\
&\forall \mathbf{g}_E\in\xi_{\mathbf{g}_E}, \forall\Delta\mathbf{H}_E\in\xi_{\mathbf{H}_E}.\label{theta2}
\end{align}
After some further manipulations, (\ref{theta1}) can be reformulated as the following semidefinite programming (SDP) problem:
\begin{align}
&{\min}_{\theta_1}\theta_1
\nonumber\\
&\mathrm{s.t.}
\left[\begin{matrix}
\bar{\mathbf{H}}_E\left(\theta_1\mathbf{\Omega}-\mathbf{Q}\right)\bar{\mathbf{H}}_E^H&
\bar{\mathbf{H}}_E\left(\theta_1\mathbf{\Omega}-\mathbf{Q}\right)
\\
\left(\theta_1\mathbf{\Omega}-\mathbf{Q}\right)\bar{\mathbf{H}}_E^H&
\theta_1\mathbf{\Omega}-\mathbf{Q}
\end{matrix}
\right]
-t\mathbf{\Phi}\succeq\mathbf{0},
\nonumber\\
&\exists t\geq 0,
\end{align}
where we have $\mathbf{\Phi}\triangleq\left[\begin{matrix}
\mathbf{I}_{N_E}&\mathbf{0}\\
\mathbf{0}&-\frac{1}{\delta^2_{\mathbf{H}_E}}
\end{matrix}
\right]$.

Similarly, after some manipulations, %the constraint of the optimization problem (\ref{theta2}) can be reformulated as
%\begin{align}
%&\theta_2\mathbf{H}_E\left(\mathbf{Q}+\mathbf{\Omega}\right)\mathbf{H}_E^H+\mathbf{I}_{N_E}\succeq P_t\mathbf{c}_E\mathbf{c}_E^H,
%\nonumber\\
%&\forall \mathbf{c}_E\in\xi_{\mathbf{c}_E}, \forall\Delta\mathbf{H}_E\in\xi_{\mathbf{H}_E}.
%\end{align}
%Employing the robust quadratic matrix
%inequality  in \cite{Wang_TVT},
(\ref{theta2}) can be reformulated as the following SDP problem:
\begin{align}
&{\min}_{\theta_2}\theta_2
\nonumber\\
&\mathrm{s.t.}
\left[\begin{matrix}
\theta_2\bar{\mathbf{H}}_E\left(\mathbf{\Omega}+\mathbf{Q}\right)\bar{\mathbf{H}}_E^H+\mathbf{I}_{N_E}-P_t\hat{\mathbf{g}}_E\hat{\mathbf{g}}_E^H&
\theta_2\bar{\mathbf{H}}_E\left(\mathbf{\Omega}+\mathbf{Q}\right)
\\
\theta_2\left(\mathbf{\Omega}+\mathbf{Q}\right)\bar{\mathbf{H}}_E^H&
\theta_2\left(\mathbf{\Omega}+\mathbf{Q}\right)
\end{matrix}
\right]
\nonumber\\
&\qquad-\eta\left[\begin{matrix}
\mathbf{I}_{N_E}&\mathbf{0}\\
\mathbf{0}&-\frac{1}{\delta^2_{\mathbf{H}_E}}
\end{matrix}
\right]\succeq\mathbf{0},
\nonumber\\
&\qquad\exists \eta\geq 0,
\end{align}
where $\hat{\mathbf{g}}_E\triangleq \bar{\mathbf{g}}_E+\frac{\bar{\mathbf{g}}_E}{||\bar{\mathbf{g}}_E||_F}\delta_{\mathbf{g}_E}$.
$R_{w}$ is then calculated by
\begin{align}
R_{w}=\mathrm{ln}\left(1+\eta_1\right)+\mathrm{ln}\left(1+\eta_2^{\min}\right)-\mathrm{ln}\left[\left(1+\theta_1^*\right)\left(1+\theta_2^*\right)\right].
\label{worstcaseratecalculation}
\end{align}

\section{Proof of Proposition 2}
For clarity, \cite[Corollary 2]{convergence} is introduced, which is given as follows.
\begin{corollary}
Consider the problem
\begin{align}
\min_{\mathbf{X},a}f\left(\mathbf{X},a\right),\quad s.t. \left(\mathbf{X},a\right)\in\mathcal{X}\times \mathcal{A},
\end{align}
where $f\left(\mathbf{X},a\right)$ is a continuously differentiable function; $\mathcal{X}\subseteq\mathbb{C}^{m\times n}$ and $\mathcal{A}\subseteq\mathbb{R}$ are closed, nonempty, and convex subsets. Suppose that the sequence $\left\{\left(\mathbf{X}^{(n)},a^{(n)}\right)\right\}$ generated by optimizing $\mathbf{X}$ and $a$ alternatively has
limit points. Every limit point of $\left\{\left(\mathbf{X}^{(n)},a^{(n)}\right)\right\}$ is a stationary point of the problem.
\end{corollary}

We have used matrix inequality for transforming the problem (\ref{ReformulatedProblem1}) into the problem (\ref{Alternatingoptimization}), equivalently, Therefore, if we prove that
the stationary point of the problem (\ref{ReformulatedProblem1}) is the KKT point of the problem (\ref{worstcaseDesign}), and the same as the stationary point of the problem (\ref{Alternatingoptimization}).
The objective function of the problem (\ref{Alternatingoptimization}) is continuously differentiable, and the feasible sets are closed and convex.
Since the objective function of the problem (\ref{Alternatingoptimization}) is nondecreasing in Algorithm 1,
due to the power constraint,
by optimizing the two variable sets: $\left(a_1,a_2.\mathbf{W}_E\right)$ and $\left(\alpha,\beta,\gamma,\mathbf{Q},\mathbf{\Omega}\right)$ alternately,
two variable sets have limit points.
Using Bolzano-Weierstrass theorem, we can conclude that $\left\{\bar{a}_1(l),\bar{a}_2(l),\bar{\alpha}(l),\bar{\beta}(l),\bar{\gamma}(l),\bar{\mathbf{W}}_E(l),\bar{\mathbf{Q}}(l),\bar{\mathbf{\Omega}}(l)\right\}$ has limit points.
Therefore, with   \cite[Corollary 2]{convergence}, the limit points generated by Algorithm 1, i.e., $\left\{\bar{a}_1^*,\bar{a}_2^*,\bar{\alpha}^*,\bar{\beta}^*,\bar{\gamma}^*,\bar{\mathbf{W}}_E^*,\bar{\mathbf{Q}}^*,\bar{\mathbf{\Omega}}^*\right\}$ is a stationary point of the problem (\ref{Alternatingoptimization}), and
$\left\{\bar{a}_1^*,\bar{a}_2^*,\bar{\mathbf{W}}_E^*,\bar{\mathbf{Q}}^*,\bar{\mathbf{\Omega}}^*\right\}$ is
a stationary point of the problem (\ref{ReformulatedProblem1}).

Next, we will show that the stationary point of the problem (\ref{ReformulatedProblem1}) is a KKT point of the problem (\ref{worstcaseDesign}).

Since with Algorithm 2, we can construct $\mathbf{Q}$ whose has one-rank without performance deterioration, the constraint $\mathrm{rank}\left(\mathbf{Q}\right)=1$ of the problem (\ref{worstcaseDesign}) and (\ref{ReformulatedProblem1}) can be omitted.
For brevity, we denote the objective function of the problem (\ref{worstcaseDesign}) as $\varphi\left(\mathbf{Q},\mathbf{\Omega}\right)$.
Since $\left\{\bar{a}_1^*,\bar{a}_2^*,\bar{\alpha}^*,\bar{\beta}^*,\bar{\gamma}^*,\bar{\mathbf{W}}_E^*,\bar{\mathbf{Q}}^*,\bar{\mathbf{\Omega}}^*\right\}$ is a stationary point of the problem (\ref{ReformulatedProblem1}), we have
\begin{align}
&\mathrm{tr}\left(\nabla_{\mathbf{Q}}\Xi\left(a_1^*,a_2^*,\mathbf{W}_E^*,\mathbf{Q}^*,\mathbf{\Omega}^*\right)^H
\left(\mathbf{Q}-\mathbf{Q}^*\right)\right)\leq 0,\nonumber\\
&\mathrm{tr}\left(\mathbf{Q}+\mathbf{\Omega}\right)\leq P_{tot}\nonumber\\
&\mathrm{tr}\left(\nabla_{\mathbf{\Omega}}\Xi\left(a_1^*,a_2^*,\mathbf{W}_E^*,\mathbf{Q}^*,\mathbf{\Omega}^*\right)^H
\left(\mathbf{\Omega}-\mathbf{\Omega}^*\right)\right)\leq 0,
\nonumber\\
&\mathrm{tr}\left(\mathbf{Q}+\mathbf{\Omega}\right)\leq P_{tot}\nonumber\\
&\mathrm{tr}\left(\nabla_{a_1}\Xi\left(a_1^*,a_2^*,\mathbf{W}_E^*,\mathbf{Q}^*,\mathbf{\Omega}^*\right)^H
\left(a_1-a_1^*\right)\right)\leq 0,a_1>0\nonumber\\
&\mathrm{tr}\left(\nabla_{a_2}\Xi\left(a_1^*,a_2^*,\mathbf{W}_E^*,\mathbf{Q}^*,\mathbf{\Omega}^*\right)^H
\left(a_2-a_2^*\right)\right)\leq 0,a_2>0\label{DestinationResult}
\end{align}
From the equivalent relation (\ref{equivalentequation1})-(\ref{equivalentequation3}), we can obtain
\begin{align}
&\mathbf{W}_E^*=\left(\mathbf{I}_{N_E}+\mathbf{H}_E\left(\mathbf{v}\mathbf{v}^H+\mathbf{\Omega}\right)\mathbf{H}_E^H+P_t\mathbf{g}_E\mathbf{g}_E^H\right)^{-1}
\label{result1}\\
&a_1=\left(P_t|g|^2+\mathbf{h}_B^H\mathbf{\Omega}\mathbf{h}_B+1\right)^{-1}
\label{result2}\\
&a_2=\left(1+\mathbf{r}^H\mathbf{H}_L\left(\mathbf{v}\mathbf{v}^H+\mathbf{\Omega}\right)\mathbf{H}_L^H\mathbf{r}\right)^{-1}
\label{result3}
\end{align}
Substituting (\ref{result1})-(\ref{result3}) into (\ref{DestinationResult}), we can obtain
\begin{align}
&\nabla_{\mathbf{Q}}\Xi\left(a_1^*,a_2^*,\mathbf{W}_E^*,\mathbf{Q}^*,\mathbf{\Omega}^*\right)
=\nabla_{\mathbf{Q}}\varphi\left(\mathbf{Q}^*,\mathbf{\Omega}^*\right)\nonumber\\
&\nabla_{\mathbf{\Omega}}\Xi\left(a_1^*,a_2^*,\mathbf{W}_E^*,\mathbf{Q}^*,\mathbf{\Omega}^*\right)
=\nabla_{\mathbf{\Omega}}\varphi\left(\mathbf{Q}^*,\mathbf{\Omega}^*\right).
\end{align}
Therefore, with the equation above, we can further obtain
\begin{align}
&\mathrm{tr}\left(\nabla_{\mathbf{Q}}\varphi\left(\mathbf{Q}^*,\mathbf{\Omega}^*\right)^H
\left(\mathbf{Q}-\mathbf{Q}^*\right)\right)\leq 0,\mathrm{tr}\left(\mathbf{Q}+\mathbf{\Omega}\right)\leq P_{tot}\nonumber\\
&\mathrm{tr}\left(\nabla_{\mathbf{\Omega}}\varphi\left(\mathbf{Q}^*,\mathbf{\Omega}^*\right)^H
\left(\mathbf{\Omega}-\mathbf{\Omega}^*\right)\right)\leq 0,\mathrm{tr}\left(\mathbf{Q}+\mathbf{\Omega}\right)\leq P_{tot}.
\end{align}
Then, we can conclude that $\mathbf{Q}^*,\mathbf{\Omega}^*$ is the optimal solution of the following problem:
\begin{align}
\max_{\mathbf{Q},\mathbf{\Omega}}\mathrm{tr}\left(\nabla_{\mathbf{Q}}\varphi\left(\mathbf{Q}^*,\mathbf{\Omega}^*\right)^H
\left(\mathbf{Q}-\mathbf{Q}^*\right)\right), s.t.\mathrm{tr}\left(\mathbf{Q}+\mathbf{\Omega}\right)\leq P_{tot}
\label{NewAddedProblem}
\end{align}

Hence, $\mathbf{Q}^*,\mathbf{\Omega}^*$ should satisfy the KKT conditions of the problem (\ref{NewAddedProblem}), i.e.,
\begin{align}
&\nabla_{\mathbf{Q}}\varphi\left(\mathbf{Q}^*,\mathbf{\Omega}^*\right)-\lambda_{\mathbf{Q}}\mathbf{I}+
\mathbf{D}_{\mathbf{Q}}=\mathbf{0}\nonumber\\
&\mathbf{Q}^*\mathbf{D}_{\mathbf{Q}}=\mathbf{0},\lambda_{\mathbf{Q}}\geq 0,\mathbf{D}_{\mathbf{Q}}\succeq\mathbf{0}\nonumber\\
&\nabla_{\mathbf{\Omega}}\varphi\left(\mathbf{Q}^*,\mathbf{\Omega}^*\right)-\lambda_{\mathbf{\Omega}}\mathbf{I}+
\mathbf{D}_{\mathbf{\Omega}}=\mathbf{0}\nonumber\\
&\mathbf{Q}^*\mathbf{D}_{\mathbf{\Omega}}=\mathbf{0},\lambda_{\mathbf{\Omega}}\geq 0,\mathbf{D}_{\mathbf{Q}}\succeq\mathbf{0},
\label{KKTcondition}
\end{align}
where $\mathbf{D}_{\mathbf{\Omega}}$,  $\lambda_{\mathbf{\Omega}}$, $\mathbf{D}_{\mathbf{Q}}$,  and $\lambda_{\mathbf{Q}}$ are Lagrangian multipliers. The condition (\ref{KKTcondition}) is exactly the KKT codition of the problem (7).
% or


\ifCLASSOPTIONcaptionsoff
  \newpage
\fi



%\bibliographystyle{IEEEtran}

\bibliographystyle{IEEEtran}

\begin{IEEEbiography}[{\includegraphics[width=1in,height=1.25in,clip,keepaspectratio]{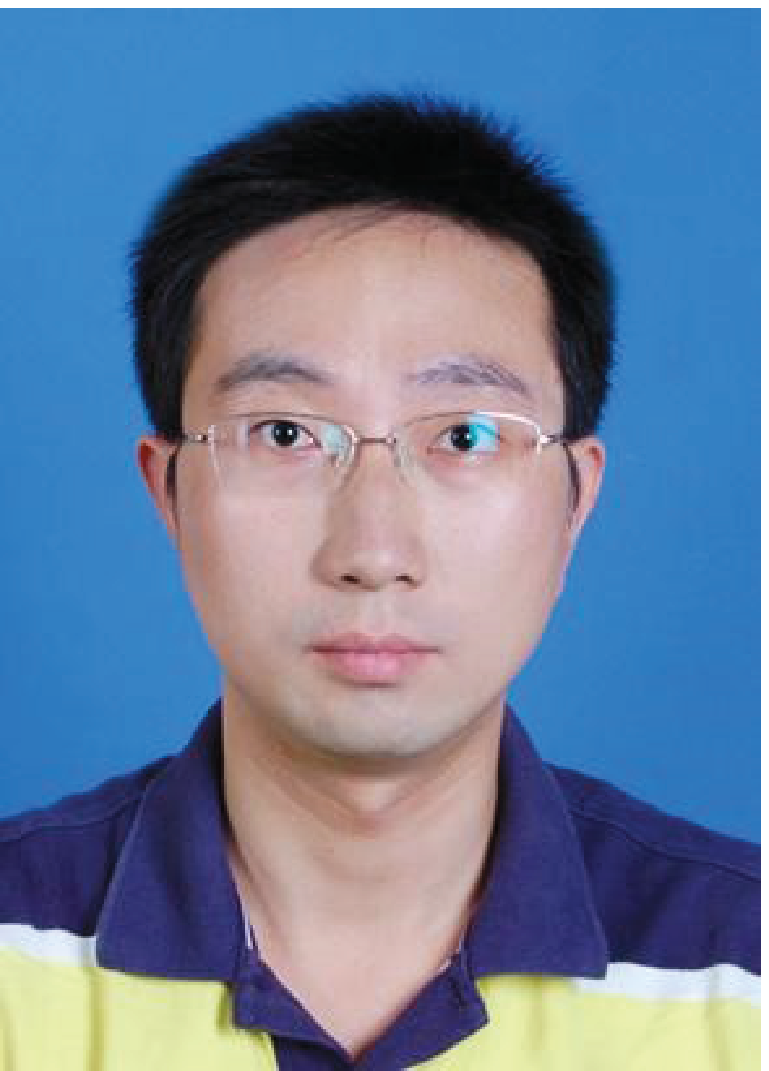}}] {Zhengmin Kong} received the B.Eng. and Ph.D. degrees from the School of Electronic Information and Communications, Huazhong University of Science and Technology, Wuhan, China, in 2003 and 2011, respectively. From 2005 to 2011, he was with the Wuhan National Laboratory for Optoelectronics as a member of the Research Staff, and was involved in Beyond-3G and UWB system design. He is currently an Associate Professor with the School of Electrical Engineering and Automation, Wuhan University, Wuhan, China. From 2014 to 2015, he was with the University of Southampton, U.K., as an Academic Visitor, and investigated physical layer security and interference management techniques. His current research interests include wireless communications, smart grid and signal processing.
\end{IEEEbiography}

\begin{IEEEbiography}[{\includegraphics[width=1in,height=1.25in,clip,keepaspectratio]{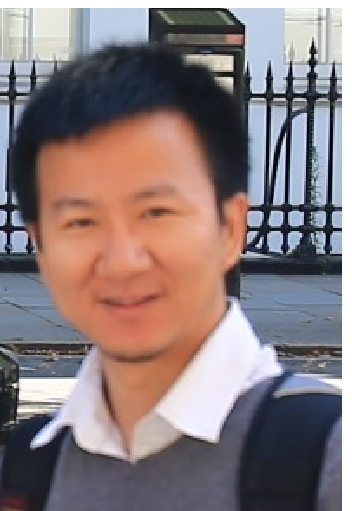}}] {Shaoshi Yang} (S'09-M'13-SM'19) received the B.Eng. degree in information engineering from Beijing University of Posts and Telecommunications (BUPT), China, in 2006 and the Ph.D. degree in electronics and electrical engineering from University of Southampton, U.K., in 2013. From 2008 to 2009, he was involved in the Mobile WiMAX standardization with Intel Labs China. From 2013 to 2016, he was a Research Fellow with the School of Electronics and Computer Science, University of Southampton. From 2016 to 2018, he was a Principal Engineer with Huawei Technologies Co., Ltd., where he led the company?s research efforts in wireless video/VR transmission. Currently, he is a Full Professor at BUPT. His research interests include large-scale MIMO signal processing, millimetre wave communications, wireless AI and wireless video/VR in 5G and beyond. He received the Dean?s Award for Early Career Research Excellence from the University of Southampton, and the President Award of Wireless Innovations from Huawei. His research excellence was also recognized by the National Thousand-Young-Talent Fellowship. He is a member of the Isaac Newton Institute for Mathematical Sciences, Cambridge University, and an Editor for IEEE Wireless Communications Letters. He was also an Associate Editor for IEEE Journal on Selected Areas in Communications, and an invited international reviewer for the Austrian Science Fund (FWF). (http://shaoshiyang.weebly.com/) 
\end{IEEEbiography}

\begin{IEEEbiography}[{\includegraphics[width=1in,height=1.25in,clip,keepaspectratio]{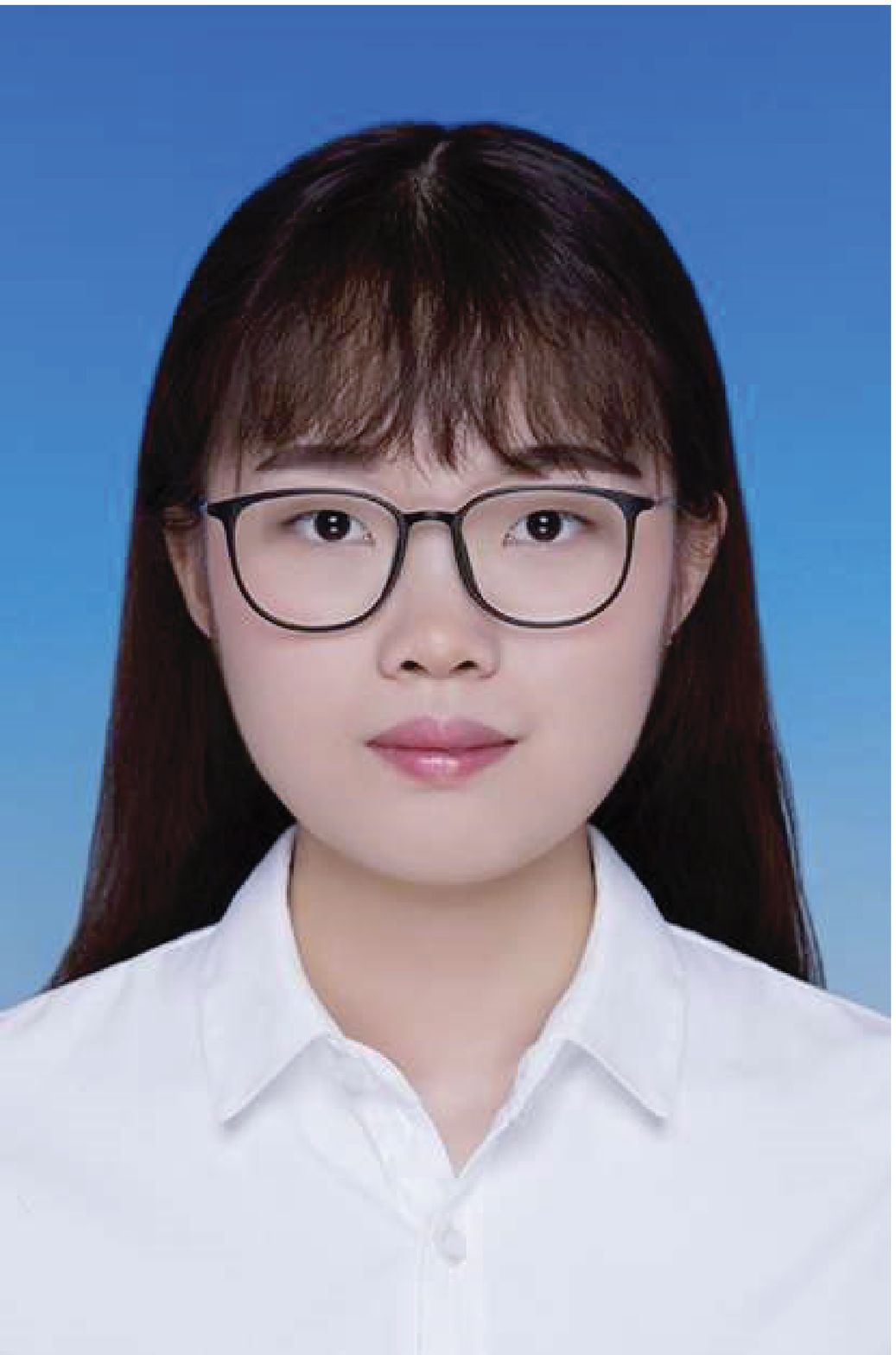}}] {Die Wang} received the B.S. degree in Automation from Wuhan University, Wuhan, China, in 2018. She is currently pursuing the M.S. degree with the School of Electrical Engineering and Automation, Wuhan University, Wuhan, China. Her current research interests include wireless communications and signal processing, interference management schemes in MIMO interference channels, capacity analysis in multiuser communication systems, physical layer security and full duplex.
\end{IEEEbiography}

\begin{IEEEbiography}[{\includegraphics[width=1in,height=1.25in,clip,keepaspectratio]{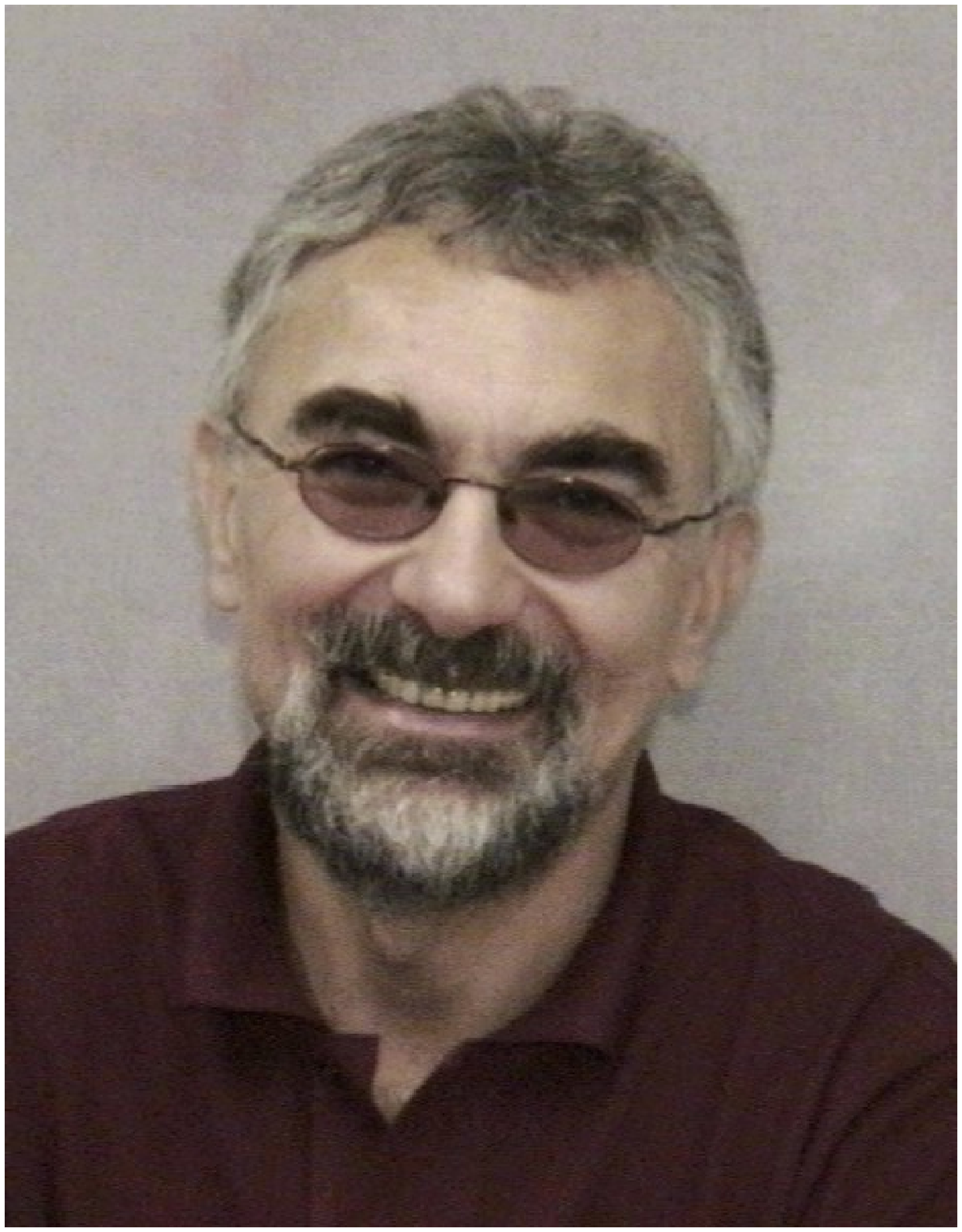}}] {Lajos Hanzo}
(M'91-SM'92-F'04) FREng, FIET, Fellow of EURASIP, received his 5-year degree in electronics in 1976 and his doctorate in 1983 from the Technical University of Budapest.  In 2009 he was awarded an honorary doctorate by the Technical University of Budapest and in 2015 by the University of Edinburgh.  In 2016 he was admitted to the Hungarian Academy of Science. During his 40-year career in telecommunications he has held various research and academic posts in Hungary, Germany and the
UK. Since 1986 he has been with the School of Electronics and Computer
Science, University of Southampton, UK, where he holds the chair in
telecommunications.  He has successfully supervised 119 PhD students,
co-authored 18 John Wiley/IEEE Press books on mobile radio
communications totalling in excess of 10 000 pages, published 1800+
research contributions at IEEE Xplore, acted both as TPC and General
Chair of IEEE conferences, presented keynote lectures and has been
awarded a number of distinctions. Currently he is directing a
60-strong academic research team, working on a range of research
projects in the field of wireless multimedia communications sponsored
by industry, the Engineering and Physical Sciences Research Council
(EPSRC) UK, the European Research Council's Advanced Fellow Grant and
the Royal Society's Wolfson Research Merit Award.  He is an
enthusiastic supporter of industrial and academic liaison and he
offers a range of industrial courses.  He is also a Governor of the
IEEE ComSoc and VTS.  He is a former Editor-in-Chief of
the IEEE Press and a former Chair Professor also at Tsinghua University,
Beijing.  For further information on research in progress and
associated publications please refer to
http://www-mobile.ecs.soton.ac.uk 
\end{IEEEbiography}

\end{document}